\documentclass[11pt, a4paper]{article}
\usepackage{amssymb,amsmath}
\usepackage{amsthm}
\usepackage{fullpage}
\usepackage{algorithm2e}
\newtheorem{theorem}{Theorem}[section]
\newtheorem{definition}{Definition}[section]

\newtheorem{lemma}{Lemma}[section]

\title{Randomness Extraction  over Bilinear Character Sums}

\vspace{-2.5cm}
\author{\\Boudjou T. Hortense,\scriptsize{\textit{Universite de Maroua-Cameroon }}; \\Dr Abdoul Aziz Ciss, \scriptsize{\textit{Ecole Polytechnique de Thies-Senegal }}}
\begin{document}
\maketitle
\begin{abstract}

This work is based on the proposal of a deterministic randomness extractor of a random Diffie-Hellman element defined over two prime order multiplicative subgroups of a finite fields $\mathbb{F}_{p^n}$, $G_1$ and $G_2$. We show that the least significant bits of a random element in $G_1*G_2$, are indistinguishable from a uniform bit-string of the same length.\\
One of the main application of this extractor is to replace the use of hash functions in pairing by the use of a good deterministic randomness extractor.\\ 


\end{abstract}

\section{Introduction}

The shared element after a Diffie-Helmann exchange is $g^{ab}\in G$, where $G$ is a cyclic subgroup of a finite field. $g^{ab}$ is indistinguishable from any other element of $G$ under the decisionnal Diffie-Hellman (DDH) assumption \cite{Boneh}. This hypothesis argues that, given two distribution $(g^a,g^b,g^{ab})$ and $(g^a,g^b,g^c)$ there is no efficient algorithm that can distinguish them. However, the encryption key should be indistinguishable from a random bit string having a uniform distribution.So we could not directly use $g^{ab}$ as the encryption key. 
It is therefore of adequate arrangements to ensure the indistinguishability of the key  such as hash functions, pseudo-random functions or random extractors. \\
Deterministic random extractor have been introduced in complexity theory by Trevisan and Vadhan \cite{Trevisan}.
Most of the work on deterministic extractors using exponential sums for their security proof work with simple exponential sums \cite{Boneh2, Chevalier, Ciss, Diffie,Fouque}. Here we introduce deterministic random extractors that extract a perfectly random bit string of an element derived from the combination of two separate source.\\

\textbf{Related work}

In 1998, Boneh et \textit{al.} \cite{Boneh2} show that calculate the $k$-\textit{most significant bits} of a secrete is also difficult as to calculate the common secret .The authors rely on \textit{Hidden Number Problem}.\\
 Hastad et \textit{al}. \cite{Hastad} propose random extractor based on the probabilistic \textit{Leftover Hash Lemma}, capable of removing all of the entropy random source having sufficient min-entropy. This technique and its variants, however, requires the use of hash functions and  perfect random.\\

The particularity of these extractors is that they belong to the random oracle model. Thus, indistinguishability can not be proven under the DDH assumption unless you add a random oracle. However, these are considered some limitations in practice.\\
In 2008, Fouque  et \textit{al}. \cite{Fouque} propose a simple extractor capable of extracting the $k$ least significant bits or the $k$ most significant bits of a strong random element issued to the Diffie-Hellman exchange on a sufficent big subgroup of $\mathbb{Z}_p$. They rely on exponential sums to bound the statistical distance between two variable.\\
In 2009, Chevalier et \textit{al}. \cite{Chevalier} also use exponential sums but bound the collision probability of  bits extracted to prove the security of the extractor.They use the \textit{Vinogradov inequality}  to limit the incomplete character sums. They improve the results of Fouque by providing an extractor capable of extracting up to two times more bits. They also feature extractor on the group of points of an elliptic curve defined over a finite field. However, their work was limited to the  finite prime fields.\\
In 2011,Ciss et \textit{al}. \cite{Ciss} extend the work of Chevalier over finite non prime fields $\mathbb{F}_{p^n}$ and elliptic curves over $\mathbb{F}_{p^n}$ and more particularly on binary finite fields. They use the \textit{Winterhof inequality}  to limit the incomplete character sums.\\

All that previous work are based on the caracter model, using single character sums. we focus on the extraction of a random string of bits from a random element from multiple source in particular, two source.\\

\textbf{Our work}

We proposed a deterministic random extractor under the DDH asumption, which maps two multiplicative subgroup of a finite field to the set $\{0,1\}^k$, permitting to extract the $k$-least significant bits of a random element issue of the two subgroup. We use the double exponential sums to bound the collision probability and give a security proof of our extractor.\\

\textbf{Organization of work}\\
This work is organize as follow: In section 2, we recall some definition and results about randomness, character sums and bilinear character sums. In section 3, we present and analyze our randomness extractor. In section 4, we finish by giving some applications of our extractor.

\section{Preliminaries}{Measures of randomness}
In this section, we introduce some definitions and results on the measurement parameters of randomness \cite{Shoup} and on character sums.\\

\subsection{Measures of randomness}

\begin{definition}{ Guessing probability}\\
Let $\mathcal{X}$ be a set of cardinality $|\mathcal{X}|$ and $X$, an $\mathcal{X}$-valued random variable.\\ 
The \textit{guessing probability} $ \gamma (X) $ of $ X $ is given by: 
\begin{center}$\displaystyle \gamma(X) = max\{P[X = v]: v \in \mathcal{X}\} $ \end{center}
\end{definition}

\begin{definition}{ collision probability }\\
Let $\mathcal{X}$ be a finite set and $X$, an $\mathcal{X}$-valued random variable. The collision probability of $X$, denoted by $Col(X)$, is the probability 
\begin{center}$ \displaystyle Col(X)=Pr[X=X']=\sum_{x\in\mathcal{X}}Pr[X=x]^2$\end{center}
\end{definition}

\begin{definition}{ Statistical distance }\\
Let $\mathcal{X}$ be a finite set. If $X$ and $Y$ are $\mathcal{X}$-valued random variables, then the statistical distance $SD(X,Y )$ between $X$ and $Y$ is defined as
\begin{center}$\displaystyle SD(X,Y)=\frac{1}{2}\sum_{x\in\mathcal{X}}|Pr[X=x]-Pr[Y=x]|$\end{center}
\end{definition}

 Let $U_\mathcal{X}$ be a random variable uniformly distributed on $\mathcal{X}$ and $\delta\leq 1$ a positive real number. Then a random variable $X$ on $\mathcal{X}$ is said to be $\delta-uniform$ if 
 $$SD(X,U_\mathcal{X})\leq\delta$$

\begin{lemma}{Relation between SD and Col(X)}\\
\label{lem1}
Let $X$ be a random variable over a finite set $\mathcal{X}$ of size $|\mathcal{X}|$ and $\Delta = SD(X,U_S)$ be the statistical distance between $X$ and $U_\mathcal{X}$, where $U_\mathcal{X}$ is a uniformly distributed random variable over $\mathcal{X}$. Then,

\begin{center}$\displaystyle Col(X)\geq\frac{1+4\Delta^2}{|\mathcal{X}|}$\end{center}
\end{lemma}

\begin{definition}{Deterministic ($\mathcal{Y},\delta$)-extractor}\\
Let $\mathcal{X}$ and $\mathcal{Y}$ be two finite sets. Let $Ext$ be a function Ext : $\mathcal{X} \to \mathcal{Y}$. We say that $Ext$ is a deterministic ($\mathcal{Y},\delta$)-extractor for $\mathcal{X}$ if $Ext(U_\mathcal{X})$ is $\delta$-uniform on $\mathcal{Y}$. That is\\
\begin{center}$\displaystyle SD(Ext(U_\mathcal{X}),U_\mathcal{Y})\leq\delta $\end{center}
\end{definition}

\begin{definition}{Two-sources-extractor}\\
Let $\mathcal{X}$, $\mathcal{Y}$ and  $\mathcal{Z}$ be finite sets. The function\\
$F: \mathcal{X}\mathrm{x}\mathcal{Y}\to\mathcal{Z}$ is a two-sources-extractor if the distribution $F(X,Y)$ is $\delta$-close to the uniform distribution $U_Z\in\mathcal{Z}$ for every uniformly distributed random variables $X\in\mathcal{X}$ and $Y\in\mathcal{Y}$
\end{definition}

\subsection{Characters}
\begin{definition}
Let $G$ be an abelian group. A character of $G$ is a homomorphism from $G \to \mathbb{C}^*$. A character is trivial if it is identically 1. We denote the trivial character by $\mathcal{X}_0$ or $\psi_0$. 
\end{definition}
\begin{definition}
 Let $\mathbb{F}_q$ be a given finite field. An additive character $\psi : \mathbb{F}_q^+ \to \mathbb{C}$ is a character $\psi$ with $\mathbb{F}_q$ considered as an additive group. A multiplicative character $\mathcal{X} : \mathbb{F}_q^* \to \mathbb{C}$ is a character with $\mathbb{F}_q^* = \mathbb{F}_q - \{0\}$ considered as a multiplicative group. We extend $\mathcal{X}$ to $\mathbb{F}_q$ by defining $\mathcal{X}(0) = 1$ if $\mathcal{X}$ is trivial, and $\mathcal{X}(0) = 0$ otherwise. Note that the extended $\mathcal{X}$ still preserves multiplication. 
\end{definition}

\subsection {Exponential sums over finite fields}

The main interests of exponential sums is that they allows to construct some caracteristic functions and in some cases we know good bounds for them. The use of these caracteristic functions can permit to evaluate the size of these sets.\\

We focus on certain character sums, those involving the character  $e_p$ define as it follows.\\ 

\begin{theorem}{Multiplicative characters of $\mathbb{F}_q$}\\
The multiplicative characters of $\mathbb{F}_q$ are given by: \\
$\displaystyle \forall x\in\mathbb{F}_q$, $e_q(x) = e^{\frac{2i\pi x}{q}}\in \mathbb{C}^*$
 \end{theorem}
 
 \begin{theorem}{Additive characters of $\mathbb{F}_q$}\\
 Suppose $q = p^r$ with $p$ prime. The additive characters of $\mathbb{F}_q$ are given by \\
 $\displaystyle\psi(x) = e_p(Tr(x))$ where $Tr(x) = x + x^p + ... + x^{p^{n-1}}$ is the trace of $x$.
 \end{theorem}

\subsubsection{Single character sums} 

 Let $p$ be a prime number, $G$ a multiplicative subgroup of $\mathbb{F}_{p}^*$ .\\
 For all $a\in \mathbb{F}_{p^*}$, let introduce the following notation:
 
\begin{center}$\displaystyle S(a,G)=\sum_{x\in G}e_p(ax)$\end{center}.

\begin{lemma} {Let} $p$ be a prime number, $G$ a multiplicative subgroup of $\mathbb{F}_{p}^*$ .\\

(1)\qquad if $a=0$, \qquad $\sum_{x=0}^{p-1}e_p(ax)=p$\\

(2)\qquad For all $a\in\mathbb{F}_p^*$,\qquad $\sum_{x=0}^{p-1}e_p(ax)=0$\\

(3)\qquad For all $x_0\in G$ and all $a\in\mathbb{F}_p^*$, $S(ax_0,G)=S(a,G)$\\
\label{lem2}
\end{lemma}
\begin{proof} Follows \cite{Zimmer}, pp69-70 \end{proof}

\begin{theorem}{Polya-Vinogradov bound}\\
Let $p$ be a prime number, $G$ a multiplicative subgroup of $\mathbb{F}_{p}^*$ .\\
For all $a\in\mathbb{F}_p^*$:
\begin{center}$\displaystyle \left|\sum_{x\in G}e_p(ax)\right|\leq \sqrt{p}$\end{center}
\end{theorem}
\begin{proof}{See \cite{Zimmer} for the proof}\end{proof}

\begin{theorem}{Winterhof bound}\\
\label{winter}
Let $V$ be an additive subgroup of $\mathbb{F}_{p^n}$ and let $\psi$ be an additive caracter of $\mathbb{F}_{p^n}$. Then\\
\begin{center}$ \displaystyle \sum_{a\in \mathbb{F}_{p^n}}\left|\sum_{x\in V}\psi(ax)\right|\leq p^n$ \end{center}
\end{theorem}
\begin{proof}{See \cite{Winterhof} for the proof}\end{proof}

\subsubsection{Bilinear character sums}

Let $p$ be a prime number, $G, H$ be two multiplicative subgroups of $\mathbb{F}_{p}^*$ .\\
 For all $a\in \mathbb{F}_{p^*}$, let introduce the following notation:
\begin{center}$\displaystyle S(a,(G,H))=\sum_{x\in G}\sum_{y\in H}e_p(axy)$\end{center}

\begin{lemma} {Let} $p$ be prime and, $G$ and $H$ two subsets of $\mathbb{F}_{p}^*$. Then\\ 
\begin{center}$\displaystyle \max_{(n,p)=1}|\sum_{x\in G}\sum_{y\in H}(e_p(nxy))|\leq (p|G||H|)^{\frac{1}{2}}$\end{center}
\label{lem3}
 \end{lemma}
\begin{proof} See \cite{Bourgain,Vinogradov}\end{proof}

\begin{lemma} {For} any subsets $G$, $H$ of $\mathbb{F}_{p^n}^*$ and for any complex coefficients $\alpha_x, \beta_y$ with $|\alpha_x|\leq 1$, $|\beta_y|\leq 1$, the following bound holds\\
\begin{center}$\displaystyle |\sum_{x\in G}\sum_{y\in H}\alpha_x \beta_y\psi(xy)|\leq (p^n|G||H|)^{\frac{1}{2}}$\end{center}
\label{lem5}
\end{lemma} 

\subsection{Exponential sums over points of elliptic curves}

\subsubsection{Elliptic curves}
Let $\mathcal{E}$ be an elliptic curve over $\mathbb{F}_p$, $p\geq 3$ defined by an affine Weieirstrass equation of the form \\
\begin{equation}
y^2=x^3+ax+b
\end{equation}
with coefficients $a,b\in\mathbb{F}_p$. It is known that the set $\mathcal{E}(\mathbb{F}_p)$ of $\mathbb{F}_p$-rational points of $\mathcal{E}$, with the point at infinity $\mathcal{O}$ as the neutral element, forms an abelian group. The group law operation is denoted by $\oplus$. Every point $\mathrm{P}\neq \mathcal{O} \in \mathcal{E}(\mathbb{F}_p)$ is denote by $\mathrm{P}=(x(\mathrm{P}),y(\mathrm{P}))$. Given an integer $n$ and a point $\mathrm{P}\in \mathcal{E}(\mathbb{F}_p)$, we write $n\mathrm{P}$ for the sum of $n$ copies of $\mathrm{P}$\\
$\displaystyle n\mathrm{P}=\mathrm{P}\oplus\mathrm{P}\oplus\hdots\oplus\mathrm{P}$, $n$ copies.\\

\subsubsection{Bilinear sums over additive character}

Given two subsets $\mathcal{P}, \mathcal{Q}$ of  $\mathcal{E}(\mathbb{F}_p)$, and arbitrary complex functions $\sigma, \mathrm{v}$ supported on $\mathcal{P}$ and  $\mathcal{Q}$ we concider the bilinear sums of additive characters.\\
\begin{center}$\displaystyle \mathrm{V}_{\sigma,\mathrm{v}}(\psi,\mathcal{P},\mathcal{Q})=\sum_{\mathrm{P}\in\mathcal{P}}\sum_{\mathrm{Q}\in\mathcal{Q}}\sigma(\mathrm{P})\mathrm{v}(\mathrm{Q})\psi(x(\mathrm{P}\oplus\mathrm{Q}))$\end{center}

\begin{lemma}{Let} $\mathcal{E}$ be an elliptic curve defined over $\mathbb{F}_q$ where $q=p^n$, with $n\geq 1$ and let \\
$\displaystyle \sum_{\mathrm{P}\in\mathcal{P}}|\sigma(\mathrm{P})|^2\leq \mathrm{R}$ and $\displaystyle \sum_{\mathrm{Q}\in\mathcal{Q}}|\mathrm{v}(\mathrm{Q})|^2\leq \mathrm{T}$\\
Then, uniformly over all nontrivial additive character $\psi$ of $\mathbb{F}_q$\\
\begin{center}$\displaystyle |\mathrm{V}_{\sigma,\mathrm{v}}(\psi,\mathcal{P},\mathcal{Q})|<<\sqrt{q\mathrm{R}\mathrm{T}}$\end{center}
\label{lem7}
\end{lemma}
\begin{proof} See  \cite{Ahmadi} \end{proof}.\\

\section{Randomness extractor}
\subsection{Randomness extractor in finite fields}
 
We propose and prove the security of a simple deterministic randomness extractor for two subgroup $G_1$ and $G_2$ of $\mathbb{F}^*_q$ where $q=p^n$, with $p$ prime and $n\geq 1$. The main theorem of this section states that the $k$-least significant bits of a random element in $(G_1,G_2)$ are close to a truly random group-element in $\{0,1\}^k$. Our approach is from the model based on caracter sums.\\

\subsubsection{Randomness extraction in  $\mathbb{F}_p$ }

Let $\mathbb{F}_p$ be a finite prime field such that $|p|=m$.\\

Let $G_1$ and $G_2$ be two multiplicative subgroup of $\mathbb{F}^*_p$ of order $q_1$ (resp.$q_2$), with $|q_1|=l_1$, $|q_2|=l_2$.\\
Let $U_{G_1}$ (resp. $U_{G_2}$) be a  random variable uniformly distributed on $G_1$ (resp.$G_2$), and $k$ a positive integer less than $m$.\\

\begin{definition}{Extractor $f_k$ on $\mathbb{F}_{p}$}\\
The extractor $f_k$ is defined as a function 

\begin{center}$$ f_k: G_1\mathrm{x}G_2 \to \{0,1\}^k$$ $$\qquad (x_1,x_2)\longmapsto lsb_k(x_1x_2)$$\end{center}
\end{definition}

The following lemma shows that $f_k$ is a good randomness extractor.\\

\begin{lemma} {Let} $p$ be a $m$-bits prime, $G_1$ and $G_2$ be two multiplicative subgroups of $\mathbb{F}^*_p$ of order $q_1$ (resp.$q_2$), we denote $|q_1|=l_1$ and $|q_2|=l_2$.\\
Let $U_{G_1}$ (resp. $U_{G_2}$) be a  random variable uniformly distributed on $G_1$ (resp.$G_2$), and $k$ a positive integer less than $m$.\\
Let $U_k$ be a random variable uniformly distributed on $\{0,1\}^k$\\
If $\Delta = SD(f_k(U_{G_1},U_{G_2}),U_k)$ then\\
\begin{center}$\displaystyle 2\Delta \leq \sqrt{\frac{2^k}{p}} + \frac{2^{\frac{k}{2}}M(\log_2(p))^\frac{1}{2}}{q_1q_2}= 2^{\frac{k+m+\log_2(m)-(l_1+l_2)}{2}}$
\end{center}
\label{lem4}
\end{lemma}

\vspace{0.5cm}
\begin{proof}

Since $f_k(x_1,x_2)=\mathrm{lsb}_k(x_1x_2)$, this means $x_1x_2=2^ka+b$ or $x'_1x'_2=2^ka'+b'$ where $0\leq a,a'\leq 2^{m-k}$ et $0\leq b,b'\leq 2^k-1$\\

Thus $x_1x_2-x'_1x'_2=2^k(a-a')+(b-b')$ . 
If $\mathrm{lsb}_k(x_1x_2)$ and $\mathrm{lsb}_k(x'_1x'_2)$ coincide then \\$x_1x_2-x'_1x'_2=2^k(a-a')$.\\
Let $u=a-a'$ thus $0\leq u\leq2^{m-k}$\\

Let us define $K=2^k$,    \qquad    $ u_0=\mathrm{msb}_{m-k}(p-1)$,\\
if $w=2^mw_m+\hdots +2^1w_1+2^0w_0$ , $z=2^{m'}z_{m'}+\hdots+2^1z_1+2^0z_0$,  and $z<w$ then \\$\mathrm{msb}_k(z)<\mathrm{msb}_k(w)$\\
Since $0\leq a,a'\leq p-1$ therefore $u\leq u_0$\\
We introduce the following notation,
$$S(a,(G_1,G_2))=\sum_{x_1\in G_1}\sum_{x_2\in G_2}e_p(ax_1x_2)$$

We construct the caracteristic function,
$\displaystyle \textbf{1}((x_1,x_2),(x'_1,x'_2),u)=\frac{1}{p}\sum_{a=0}^{p-1}e_p(a(x_1x_2-x'_1x'_2-Ku))$, by  properties $(1)$ and $(2)$ of Lemma \ref{lem2}.

which is equal to $1$ if $x_1x_2-x'_1x'_2= Ku \mod(p)$ and $0$ otherwise. Therefore, we can evaluate $Col(f_k(U_{G_1},U_{G_2}))$ where $U_{G_1}$ (resp. $U_{G_2}$) is uniformly distributed in $G_1$ (resp. in $G_2$):\\

$\displaystyle Col(f_k(U_{G_1},U_{G_2}))$\\
$\displaystyle =\frac{1}{(q_1q_2)^2}|\{((x_1,x_2),(x'_1,x'_2))\in (G_1,G_2)^2 \exists u\leq u_0, x_1x_2-x'_1x'_2=Ku \mod(p)\}|$\\
$\displaystyle= \frac{1}{(q_1q_2)^2p}\sum_{(x_1,x_2)\in (G_1,G_2)}\sum_{(x'_1,x'_2)\in (G_1,G_2)}\sum_{u=0}^{u_0}\sum_{a=0}^{p-1}e_p(a(x_1x_2-x'_1x'_2-Ku))$\\
Then we manipulate the sums, separate some terms (a = 0) and obtain:\\
For $a=0$,\\
$\displaystyle Col(f_k(U_{G_1},U_{G_2}))=\frac{1}{(q_1q_2)^2p}\sum_{a=0}^{p-1}\sum_{(x_1,x_2)\in (G_1,G_2)}\sum_{(x'_1,x'_2)\in (G_1,G_2)}\sum_{u=0}^{u_0}e_p(0)=\frac{u_0+1}{p}$\quad (*)\\ 
For $a\in \mathbb{F}_p^*$,\\
$\displaystyle Col(f_k(U_{G_1},U_{G_2})) = \frac{1}{(q_1q_2)^2p}\sum_{a=1}^{p-1}\sum_{(x_1,x_2)\in (G_1,G_2)}\sum_{(x'_1,x'_2)\in (G_1,G_2)}\sum_{u=0}^{u_0}e_p(a(x_1x_2-x'_1x'_2-Ku))$\\
$\displaystyle= \frac{1}{(q_1q_2)^2p}\sum_{a=1}^{p-1}\sum_{(x_1,x_2)\in (G_1,G_2)}e_p(ax_1x_2)\sum_{(x'_1x'_2)\in (G_1,G_2)}e_p(-ax'_1x'_2)\sum_{u=0}^{u_0}e_p(-aKu)$\\
$\displaystyle= \frac{1}{(q_1q_2)^2p}\sum_{a=1}^{p-1}S(a,(G_1,G_2))S(-a,(G_1,G_2))\sum_{u=0}^{u_0}e_p(-aKu)$\\
$\displaystyle= \frac{1}{(q_1q_2)^2p}\sum_{a=1}^{p-1}|S(a,(G_1,G_2))|^2\sum_{u=0}^{u_0}e_p(-aKu)$\\
We inject the result of (*) then,\\
$\displaystyle Col(f_k(U_{G_1},U_{G_2}))= \frac{u_0+1}{p}+\frac{1}{(q_1q_2)^2p}\sum_{a=1}^{p-1}|S(a,(G_1,G_2))|^2\sum_{u=0}^{u_0}e_p(-aKu)$\\

We have\\
$\displaystyle \sum_{a=1}^{p-1}\sum_{u=0}^{u_0}e_p(-aKu)$\\
$\displaystyle=\sum_{a=1}^{p-1}\sum_{u=0}^{u_0}e_p(-au)$, it comes from a change of variable ($a'=Ka=2^ka \mod(p)$, with $\gcd(2,p)=1$).\\
$\displaystyle=\sum_{a=1}^{p-1}\frac{1-e_p(-a(u_0+1))}{1-e_p(-a)} $, considere the fact that $[0,u_0]$ is an interval, the sum is the geometric sum.\\
$\displaystyle=\sum_{a=1}^{p-1}\frac{\sin(\frac{\pi a(u_o+1)}{p})}{\sin(\frac{\pi a}{p})}$
$\displaystyle=2\sum_{a=1}^{\frac{p-1}{2}}\frac{\sin(\frac{\pi a(u_o+1)}{p})}{\sin(\frac{\pi a}{p})}$\\
$\displaystyle\leq 2\sum_{a=1}^{\frac{p-1}{2}}\frac{1}{\sin(\frac{\pi a}{p})}$
$\displaystyle\leq 2\sum_{a=1}^{\frac{p-1}{2}}|\frac{p}{a}|$
$\displaystyle\leq p\log_2(p)$\\

Therefore\\

 $\displaystyle Col(f_k(U_{G_1},U_{G_2})) \leq \frac{u_0+1}{p} +\frac{1}{(q_1q_2)^2p}|S(a,(G_1,G_2))|^2p\log_2(p)$\\
 $\displaystyle  \leq \frac{u_0+1}{p} + \frac{1}{(q_1q_2)^2p}(pq_1q_2p\log_2(p))$  , by Lemma \ref{lem3}\\
$\displaystyle\leq \frac{1}{p}+\frac{p\log_2(p)}{q_1q_2} $\\

We now use the Lemma \ref{lem1} which gives a relation between the statistical distance $\Delta$ of $f_k(U_{G_1},U_{G_2})$ with the uniform distribution and the collision probability: \\ $Col(f_k(U_{G_1},U_{G_2})) = \frac{1+4\Delta^2}{2^k}$ . The previous upper bound, combined with some manipulations, gives: \\

$\displaystyle 2\Delta \leq \sqrt{2^k.Col(f_k(U_{G_1},U_{G_2}))-1}\leq \sqrt{\frac{2^k}{p}} + \sqrt{\frac{2^kp(\log_2(p))}{q_1q_2}}\leq 2^{\frac{k+m+\log_2(m)-(l_1+l_2)}{2}}$
\end{proof}

\subsubsection{Randomness extraction in  $\mathbb{F}_{p^n}$ }

Consider the finite field $\mathbb{F}_{p^n}$, where $p$ is prime and $n$ is a positive integer greather than $1$.\\ 
$\mathbb{F}_{p^n}$ is a $n$-dimensional vector space over $\mathbb{F}_p$. Let $\{\alpha_1,\alpha_2,\hdots,\alpha_n\}$ be a basis of $\mathbb{F}_{p^n}$ over $\mathbb{F}_p$. That means, every element $x$ and $y$ in $\mathbb{F}_{p^n}$ can be represented in the form\\ $x = x_1\alpha_1 + x_2\alpha_2+ \ldots + x_n\alpha_n$, et $x' = x'_1\alpha_1 + x'_2\alpha_2+ \ldots + x'_n\alpha_n$.  where $x_i$ (resp. $x'_i$) $\in \mathbb{F}_{p^n}$. \\
Let $G_1$ and $G_2$ be two multiplicative subgroups of $\mathbb{F}^*_{p^n}$ of order $q_1$ (resp.$q_2$), we denote $|q_1|=l_1$, $|q_2|=l_2$.\\
Let $U_{G_1}$ (resp. $U_{G_2}$) be a  random variable uniformly distributed on $G_1$ (resp.$G_2$), and $k$ a positive integer less than $n$.\\

\begin{definition}{Extractor $F_k$ on $\mathbb{F}_{p^n}$}\\
The extractor $F_k$ is defined as a function 

$$F_k: G_1\mathrm{x}G_2 \to \{0,1\}^k$$ $$\qquad (x,x')\longmapsto (x_1x'_1,x_2x'_2,\hdots,x_kx'_k)$$
\end{definition}
The following lemma shows that $F_k$ is a good randomness extractor.\\

\begin{lemma} {Let} $p$ be a $m$-bits prime.
Let $G_1$ and $G_2$ be two multiplicative subgroups of $\mathbb{F}^*_{p^n}$ of order $q_1$ (resp.$q_2$), we denote $|q_1|=l_1$, $|q_2|=l_2$.\\
Let $U_{G_1}$ (resp. $U_{G_2}$) be a  random variable uniformly distributed on $G_1$ (resp.$G_2$), and $k$ a positive integer less than $m$.
Let $U_k$ be a random variable uniformly distributed on $\{0,1\}^k$\\
If $\Delta = SD(F_k(U_{G_1},U_{G_2}),U_k)$ then
\begin{center}$\displaystyle \Delta \leq \sqrt{\frac{p^{n+k-2}}{q_1q_2}}= 2^{\frac{km+nm-(l_1+l_2+2)}{2}}$\end{center}
\label{lem6}
\end{lemma}
 
 \vspace{0.5cm}
\begin{proof}

Let $(x,x'),(y,z)\in ( G_1,G_2)^2$\\
Let us introduce the notation\\
$\displaystyle T(a,(G_1,G_2))=\sum_{x\in G_1}\sum_{x'\in G_2}\psi(axx')$\\
Let us define the following sets\\

$R=\displaystyle\{x_{k+1}x'_{k+1}\alpha_{k+1} + x_{k+2}x'_{k+2}\alpha_{k+2}\hdots+ x_nx'_n\alpha_n\}$ , a subgroup of $\mathbb{F}_{p^n}$\\

$C=\displaystyle\{((x,x'),(y,z))\in (G_1,G_2)^2/\exists r\in R, xx'-yz=r\}$\\

$|C|=\displaystyle\frac{1}{p^n} \sum_{{x\in G_1},{x'\in G_2}}\sum_{{y\in G_1},{z\in G_2}}\sum_{r\in R}\sum_{a\in\mathbb{F}_{p^n} }\psi(a(xx'-yz-r))$\\

we can evaluate the collision probability:\\

$Col(F_k(U_{G_1},U_{G_2}))=\frac{|C|}{|G_1\mathrm{x}G_2|^2}$

\begin{center}$\displaystyle = \frac{1}{(q_1q_2)^2p^n} \sum_{(x,x')\in (G_1,G_2)}\sum_{(y,z)\in (G_1,G_2)}\sum_{r\in R}\sum_{a\in\mathbb{F}_{p^n} }\psi(a(xx'-yz-r))$\end{center}
\begin{center}$\displaystyle = \frac{1}{(q_1q_2)^2p^n} \sum_{a\in\mathbb{F}_{p^n} }\sum_{(x,x')\in (G_1,G_2)}\psi(axx')\sum_{(y,z)\in (G_1,G_2)}\psi(-ayz)\sum_{r\in R}\psi(-ar)$\end{center}
Then we manipulate the sums, separate some terms (a = 0) and obtain:\\
For $a=0$\\
$\displaystyle Col(F_k(U_{G_1},U_{G_2}))=\frac{1}{(q_1q_2)^2p^n} \sum_{a\in\mathbb{F}_{p^n}}\sum_{(x,x')\in (G_1,G_2)}\sum_{(y,z)\in (G_1,G_2)}\sum_{r\in R}\psi(0)=\frac{1}{p^k}$\\
For $a\in\mathbb{F}_{p^n}^*$\\
$\displaystyle Col(F_k(U_{G_1},U_{G_2}))=\frac{1}{(q_1q_2)^2p^n} \sum_{a\in\mathbb{F}_{p^n}^* }\sum_{(x,x')\in (G_1,G_2)}\psi(axx')\sum_{(y,z)\in (G_1,G_2)}\psi(-ayz)\sum_{r\in R}\psi(-ar)$\\
Then for all $a\in\mathbb{F}_{p^n}$\\
$\displaystyle Col(F_k(U_{G_1},U_{G_2}))=\frac{1}{p^k}+\frac{1}{(q_1q_2)^2p^n} \sum_{a\in\mathbb{F}_{p^n}^* }\sum_{(x,x')\in (G_1,G_2)}\psi(axx')\sum_{(y,z)\in (G_1,G_2)}\psi(-ayz)\sum_{r\in R}\psi(-ar)$\\
$\displaystyle Col(F_k(U_{G_1},U_{G_2}))= \frac{1}{p^k}+\frac{1}{(q_1q_2)^2p^n}\sum_{a\in\mathbb{F}^*_{p^n}}|T(a,(G_1,G_2))|^2\sum_{r\in R}\psi(-ar)$\\
$\displaystyle Col(F_k(U_{G_1},U_{G_2}))\leq \frac{1}{p^k}+\frac{p^n(q_1q_2)p^n}{(q_1q_2)^2p^n}$ , by Lemma \ref{lem5} and Theorem \ref{winter}\\
$\displaystyle Col(F_k(U_{G_1},U_{G_2}))\leq \frac{1}{p^k}+\frac{p^n}{(q_1q_2)}$\\

We now use the Lemma \ref{lem1} which gives a relation between the statistical distance $\Delta$ of $F_k(U_{G_1},U_{G_2})$ with the uniform distribution $U_k$ and the collision probability:\\ 

$Col(F_k(U_{G_1},U_{G_2})) = \frac{1+4\Delta^2}{2^k}$ . \\
$\displaystyle 2\Delta \leq \sqrt{2^k.Col(F_k(U_{G_1},U_{G_2}))-1}$\\
$\displaystyle \Delta\leq \sqrt{\frac{p^{n+k}}{4q_1q_2}}\leq \sqrt{\frac{p^{n+k}}{2^2q_1q_2}}$\\

$\displaystyle \Delta\leq \sqrt{\frac{p^{n+k-2}}{q_1q_2}}$\\
Therefore with some manipulations, we obtain the expected result:\\
$\displaystyle \Delta\leq \sqrt{\frac{p^{n+k-2}}{q_1q_2}}=2^{\frac{km+nm-(l_1+l_2+2)}{2}}$
\end{proof}

\subsection{Randomness extraction in elliptic curves}

\subsubsection{Randomness extractor in $\mathcal{E}(\mathbb{F}_p)$}
\begin{definition}
Let $p$ be a prime greater than 5. Let $\mathcal{E}$ be an elliptic curve over the finite field $\mathbb{F}_p$ and let $\mathcal{P},\mathcal{Q}$ be two subgroups of $\mathcal{E}(\mathbb{F}_p)$. Let denote $|\mathcal{P}|=q_1$ and $|\mathcal{Q}|=q_2$.\\
Then is define the function

\begin{center}$$ extrac_k: \mathcal{P}\mathrm{x}\mathcal{Q} \to \{0,1\}^k$$ $$\qquad\qquad\qquad\qquad (\mathrm{P},\mathrm{Q})\longmapsto lsb_k(x(\mathrm{P}).x(\mathrm{Q}))$$\end{center}
\end{definition}

\begin{lemma}{We now show an equivalent of Lemma \ref{lem4}}\\
\label{lemEC}
Let $\mathcal{E}$ be an elliptic curve over the finite field $\mathbb{F}_p$ and let $\mathcal{P},\mathcal{Q}$ be two subgroups of $\mathcal{E}(\mathbb{F}_p)$. Let denote $|\mathcal{P}|=q_1$ and $|\mathcal{Q}|=q_2$. Let $U_\mathcal{P}$ and $U_\mathcal{Q}$ be two random variables uniformly distributed in $\mathcal{P}$ and $\mathcal{Q}$ respectively. Let $U_k$ be the uniform distribution in $\{0,1\}^k$. Then\\
\begin{center}$\displaystyle \Delta(extrac_k(U_\mathcal{P},U_\mathcal{Q}),U_k)<<\sqrt{\frac{2^{k-2}p\log_2(p)}{q_1q_2}}=2^{\frac{k+n+\log_2(n)-(l_1+l_2+2)}{2}}$\end{center}
\end{lemma}

\begin{proof}
Let us define $K=2^k$,    \qquad    $ u_0=\mathrm{msb}_{m-k}(p-1)$\\
Define the characteristic function\\
\textbf{1}$\displaystyle ((\mathrm{P},\mathrm{Q}),(\mathrm{A},\mathrm{B}),u)=\frac{1}{p}\sum_{\psi\in\Psi}\psi(x(\mathrm{P})x(\mathrm{Q})-x(\mathrm{A})x(\mathrm{B})-Ku)$ which is equal to 1 if $\psi=\psi_0$ and to $0$, otherwise.\\

Let us compute the collision probablity\\
$\displaystyle Col(extrac_k(U_\mathcal{P},U_\mathcal{Q}))=\frac{1}{(q_1q_2)^2p}\sum_{\mathrm{P}\in\mathcal{P}}\sum_{\mathrm{Q}\in\mathcal{Q}}\sum_{\mathrm{A}\in\mathcal{P}}\sum_{\mathrm{B}\in\mathcal{Q}}\sum_{\psi\in\Psi}\sum_{u\leq u_0}\psi(x(\mathrm{P})x(\mathrm{Q})-x(\mathrm{A})x(\mathrm{B})-Ku)$\\
Then we manipulate the sums, separate some terms $(\psi = \psi_0)$ and obtain:\\
$\displaystyle Col(extrac_k(U_\mathcal{P},U_\mathcal{Q}))=\frac{1}{(q_1q_2)^2p}\sum_{\psi\in\Psi}\sum_{\mathrm{P}\in\mathcal{P}}\sum_{\mathrm{Q}\in\mathcal{Q}}\sum_{\mathrm{A}\in\mathcal{P}}\sum_{\mathrm{B}\in\mathcal{Q}}\sum_{u\leq u_0}\psi(x(\mathrm{P})x(\mathrm{Q})-x(\mathrm{A})x(\mathrm{B})-Ku)$\\
For $(\psi = \psi_0)$,\\
$\displaystyle Col(extrac_k(U_\mathcal{P},U_\mathcal{Q}))=\frac{1}{(q_1q_2)^2p}\sum_{\psi=\psi_0}\sum_{\mathrm{P}\in\mathcal{P}}\sum_{\mathrm{Q}\in\mathcal{Q}}\sum_{\mathrm{A}\in\mathcal{P}}\sum_{\mathrm{B}\in\mathcal{Q}}\sum_{u\leq u_0}\psi_0(0)$\\
$\displaystyle =\frac{1}{(q_1q_2)^2p}\sum_{\psi=\psi_0}\sum_{\mathrm{P}\in\mathcal{P}}\sum_{\mathrm{Q}\in\mathcal{Q}}\sum_{\mathrm{A}\in\mathcal{P}}\sum_{\mathrm{B}\in\mathcal{Q}}\sum_{u\leq u_0}e_p(Tr(0))$\\
$\displaystyle =\frac{1}{(q_1q_2)^2p}\sum_{\psi=\psi_0}\sum_{\mathrm{P}\in\mathcal{P}}\sum_{\mathrm{Q}\in\mathcal{Q}}\sum_{\mathrm{A}\in\mathcal{P}}\sum_{\mathrm{B}\in\mathcal{Q}}\sum_{u\leq u_0}1$\\
$\displaystyle =\frac{u_0+1}{p}$\\

For $(\psi \neq \psi_0)$,\\
$\displaystyle Col(extrac_k(U_\mathcal{P},U_\mathcal{Q}))=\frac{1}{(q_1q_2)^2p}\sum_{\psi\neq\psi_0}\sum_{\mathrm{P}\in\mathcal{P}}\sum_{\mathrm{Q}\in\mathcal{Q}}\sum_{\mathrm{A}\in\mathcal{P}}\sum_{\mathrm{B}\in\mathcal{Q}}\sum_{u\leq u_0}\psi(x(\mathrm{P})x(\mathrm{Q})-x(\mathrm{A})x(\mathrm{B})-Ku)$\\
Then \\
$\displaystyle Col(extrac_k(U_\mathcal{P},U_\mathcal{Q}))=\frac{u_0+1}{p}+\frac{1}{(q_1q_2)^2p}\sum_{\psi\neq\psi_0}\sum_{\mathrm{P}\in\mathcal{P}}\sum_{\mathrm{Q}\in\mathcal{Q}}\sum_{\mathrm{A}\in\mathcal{P}}\sum_{\mathrm{B}\in\mathcal{Q}}\sum_{u\leq u_0}\psi(x(\mathrm{P})x(\mathrm{Q})-x(\mathrm{A})x(\mathrm{B})-Ku)$\\
$\displaystyle =\frac{u_0+1}{p}+\frac{1}{(q_1q_2)^2p}\sum_{\psi\neq\psi_0}\sum_{\mathrm{P}\in\mathcal{P}}\sum_{\mathrm{Q}\in\mathcal{Q}}\psi(x(\mathrm{P})x(\mathrm{Q}))\sum_{\mathrm{A}\in\mathcal{P}}\sum_{\mathrm{B}\in\mathcal{Q}}\psi(-x(\mathrm{A})x(\mathrm{B}))\sum_{u\leq u_0}\psi(-Ku)$\\
$\displaystyle =\frac{u_0+1}{p}+\frac{1}{(q_1q_2)^2p}\sum_{\psi\neq\psi_0}|\sum_{\mathrm{P}\in\mathcal{P}}\sum_{\mathrm{Q}\in\mathcal{Q}}\psi(x(\mathrm{P})x(\mathrm{Q}))||\sum_{\mathrm{A}\in\mathcal{P}}\sum_{\mathrm{B}\in\mathcal{Q}}\psi(-x(\mathrm{A})x(\mathrm{B}))|\sum_{u\leq u_0}\psi(-Ku)$\\
$\displaystyle =\frac{u_0+1}{p}+\frac{1}{(q_1q_2)^2p}\sum_{\psi\neq\psi_0}|\mathrm{V}(\psi,\mathcal{P},\mathcal{Q})|^2\sum_{u\leq u_0}\psi(-Ku)$\\
$\displaystyle \leq\frac{1}{p}+\frac{1}{(q_1q_2)^2p}\sum_{\psi\neq\psi_0}q_1q_2p\sum_{u\leq u_0}\psi(-Ku)$, by Lemma \ref{lem7}\\
$\displaystyle \leq\frac{1}{p}+\frac{1}{(q_1q_2)^2p}pq_1q_2p\log_2(p)$,  since it is shown that $\displaystyle\sum_{\psi\neq\psi_0}\sum_{u\leq u_0}\psi(-Ku)\leq p\log_2(p)$\\
$\displaystyle \leq\frac{1}{p}+\frac{1}{(q_1q_2)}p\log_2(p)$\\
We now use the Lemma \ref{lem1}\\
$\displaystyle 2\Delta(extrac_k(U_\mathcal{P},U_\mathcal{Q}),U_k)<< \sqrt{2^k.Col(F_k(U_{G_1},U_{G_2}))-1}$\\
$\displaystyle 2\Delta(extrac_k(U_\mathcal{P},U_\mathcal{Q}),U_k)<< \sqrt{2^k(\frac{1}{p}+\frac{1}{(q_1q_2)}p\log_2(p)-1)}$\\
Therefore with some manipulations,\\
$\displaystyle \Delta(extrac_k(U_\mathcal{P},U_\mathcal{Q}),U_k)<<\sqrt{\frac{2^{k-2}p\log_2(p)}{q_1q_2}}=2^{\frac{k+n+\log_2(n)-(l_1+l_2+2)}{2}}$

\end{proof}

\subsubsection{Randomness extractor in $\mathcal{E}(\mathbb{F}_{p^n})$}
\begin{definition}
Let $p$ be a prime, $p>5$. Let $\mathcal{E}$ be an elliptic curve over the finite field $\mathbb{F}_{p^n}$. let $\mathcal{P},\mathcal{Q}$ be two subgroups of $\mathcal{E}(\mathbb{F}_{p^n})$. Let denote $|\mathcal{P}|=q_1$ and $|\mathcal{Q}|=q_2$.\\
Then is define the function

\begin{center}$$ Extrac_k: \mathcal{P}\mathrm{x}\mathcal{Q} \to \{0,1\}^k$$ $$\qquad\qquad\qquad\qquad (\mathrm{P},\mathrm{Q})\longmapsto lsb_k(x(\mathrm{P}).x(\mathrm{Q}))$$\end{center}
Where $x(\mathrm{P}).x(\mathrm{Q})= t_1\alpha_1+t_2\alpha_2+t\hdots+ t_n\alpha_n$
\end{definition}
\begin{lemma}
Let $\mathcal{E}$ be an elliptic curve over the finite field $\mathbb{F}_{p^n}$ and let $\mathcal{P},\mathcal{Q}$ be two subgroups of $\mathcal{E}(\mathbb{F}_{p^n})$. Let denote $|\mathcal{P}|=q_1$ and $|\mathcal{Q}|=q_2$. Let $U_\mathcal{P}$ and $U_\mathcal{Q}$ be two random variables uniformly distributed in $\mathcal{P}$ and $\mathcal{Q}$ respectively. Let $U_k$ be the uniform distribution in $\{0,1\}^k$. Then\\
\begin{center}$\displaystyle \Delta(Extrac_k(U_\mathcal{P},U_\mathcal{Q}),U_k)<<\sqrt{\frac{p^{n+k}}{4q_1q_2}}=2^{\frac{km+nm-(l_1+l_2+2)}{2}}$\end{center}

\end{lemma}
\begin{proof}
Using Lemma \ref{lem7} and Theorem \ref{winter}, the sketch of the proof is the same as those of Lemma \ref{lem6}
\end{proof}

\section{Application}
The first most well-known and use tools for the extraction phase of a key exchange protocol in order to create a secure chanal are hash function. Hash functions are the most aften adopted solution because of their flexibility and efficiency. However, they have a significant drawback. That is, the validity of this technique holds in the random oracle model only.\\
Definitely the truncation of the bit-string of the random element is the most efficient randomness extractor, since it is deterministic and does not require any computation.\\
The interest of studying randomness extraction has several cryptographic applications specially the randomness extraction from a point of elliptic curve. Some of these various applications are find as we have already said in key derivation function, key exchange protocols\cite {Diffie}, design cryptographically secure pseudorandom number generator\cite {Trevisan2}.\\
Today the trend is towards cryptography identification and pairing on elliptic and hyperelliptic curves are widely used in this field, especially for key exchange between three entities and for authentication. Nevertheless, we find that the tools used in most of the protocols based on the pairing, in this case for authentication using hash functions in the extraction phase. The extractor on two sources would be good candidates to replace these functions. They are not only deterministic but also offer the possibility of increasing the randomness considering either one but two sources.


\begin{thebibliography} {MM}
{\small


\bibitem{Ahmadi} O. Ahmadi, and I. E. Shparlinski. Exponential Sums over Points of Elliptic Curves. arXiv preprint arXiv:1302.4210. (2013) 

\bibitem{Balog} A. Balog, K. A. Broughan and I. E. Shparlinski. \textit{Sum-Products Estimates with Several Sets and Applications}


\bibitem{Bellare} M. Bellare and P. Rogaway. \textit{Random oracles are practical : A Paradigm for designing efficient protocols}. In V. Ashby, editor, ACM CCS 93, pages 62-73. ACM Press, Nov. 1993.


\bibitem{Boneh} D. Boneh. The decision Diffie-Hellman problem. \textit{In Third Algorithmic Number Theory Symposium (ANTS)}, vol.1423 of LNCS. Springer, 1998


\bibitem{Boneh2} D. Boneh and R. Venkatesan. \textit{Hardness of computing the most significant bits of secret keys in Diffie-Helman and related schemes}. In N. Koblitz, editor, CRYPTO'96, vol. 1109 of LNCS, pages 129-142. Springer, Aug. 1996.


\bibitem{Bourgain} J. Bourgain and M. Z. Garaev. \textit{On a variant of sum-product estimate and explicit exponential sum bounds in prime field}, Math.Proc.Camb.Phil.Soc, 146(2008), 1-21.


\bibitem{Bourgain2} J. Bourgain and S. V. Konyagin. \textit{Estimates for the Number of Sums and Products and for Exponential Sums Over Subgroups in Fields of Prime Order}. 


\bibitem{Carnetti} R. Carneti, J. Friedlander, S. Koyagin, M. Larsen, D. Lieman and I. Shparlinski. \textit{On the Statistical Properties of Diffie-Hellman Distributions}. Israel Journal of Mathematics, vol. 120, pages 23-46, 2000.


\bibitem{Carnetti2} R. Carnetti, J. Friedlander, and I. Shparlinski. \textit{On Certain Exponential Sums and the Distribution of Diffie-Hellman Triples}. Journal of the London Mathematical Society, 59(2):799-812, 1999. 


\bibitem{Chevalier} C. Chevalier, P. Fouque, D. Pointcheval and S. Zimmer, \textit{Optimal Randomness Extraction from a Diffie-Hellman Element, Advances in Cryptology}- Eurocrypt'09, vol. 5479 of LNCS, pages 572-589, Springer-Verlag, 2009


\bibitem{Ciss} A. A. Ciss and D. Sow. \textit{On Randomness Extraction in Elliptic Curves}. In A. Nitaj and D. Pointcheval, editors. Africacrypt 2011, vol. 6737 of LNCS, pages 290-297. Springer-Verlag, 2011.


\bibitem{Diffie} W. Diffie, M. Hellman, \textit{New Directions in Cryptography, IEEE Trans- actions On Information Theory}, vol.22, no.6, 644-654, 1976


\bibitem{Fouque} P. A. Fouque, D. Pointcheval, J. Stern, and S. Zimmer. \textit{Hardness of distinguishing the MSB or the LSB of secret keys in Diffie-Hellman schemes}. In M. Bugliesi, B. Preneel, V. Sassone, and I. Wegener, editors, ICALP 2006, Part II, vol. 4052 of LNCS, pages 240-251. ACM, 2008.


\bibitem{Hastad} J. Hästad, R. Impagliazzo, L. Levin, and M. Luby, \textit{A pseudorandom generator from any one-way function}, SIAM Journal on Computing, Vol. 28, no.4, 1364-1396,1999


\bibitem{Koyagin} S. V. Koyagin and I. Shparlinski. \textit{Character Sums With Exponential Functions and Their Applications}. Cambridge University Press, Cam- 
bridge, 1999.

\bibitem{Trevisan2} L. Trevisan. \textit{Extractors and pseudorandom generators}. J. ACM 48, 4 (July 2001), 860-879, (2001). 

\bibitem{Trevisan} L. Trevisan and S. Vadhan, \textit{Extracting Randomness from Samplable Distributions}, IEEE Symposium on Foundations of Computer Science, 32-42, 2000


\bibitem{Shoup} V. Shoup \textit{A Computational Introduction to Number Theory and Algebra Cambridge University Press}, Cambridge 2005.



\bibitem{Vinogradov} I. M. Vinogradov. \textit{An Introduction to the Theory of Numbers} (Pergamon Press, 1955).


\bibitem{Winterhof} A. Winterhof. \textit{Incomplete Additive Character Sums and Applications}. In D. Jungnickel and H. Niederreiter, editors. Finite Fields and Applications, pages 462-474. Springer-Velag 2001.


\bibitem{Zimmer} S. Zimmer "`Mécanismes cryptographiques pour la génération de cléfs et l'authentification"',

  }

\end{thebibliography}
\end{document}